\documentclass[sigconf,authorversion]{acmart}
\pdfoutput=1

\usepackage[utf8]{inputenc}
\usepackage[T1]{fontenc}
\usepackage{microtype}

\usepackage{formatting/shortcuts}
\usepackage{amsthm}
\usepackage{tabularx}
\usepackage{enumitem}
\DeclareMathDelimiter{(}{\mathopen} {operators}{"28}{largesymbols}{"00}
\DeclareMathDelimiter{)}{\mathclose}{operators}{"29}{largesymbols}{"01}

\newtheorem{theorem}{Theorem}[section]

\usepackage{pgfplotstable}
    \pgfplotsset{
        compat=1.7,
    }
    \pgfplotstableread[col sep=comma]{data/vsa_bound_all.csv}{\loadedtable}
        \pgfplotstablegetcolsof{\loadedtable}
        \pgfmathtruncatemacro{\NoOfCols}{\pgfplotsretval-1}

\algtext*{EndWhile}% Remove "end while" text
\algtext*{EndIf}% Remove "end if" text
\algtext*{EndProcedure}% Remove "end if" text

\setcopyright{rightsretained}
\acmPrice{}
\acmDOI{10.1145/3468264.3468621}
\acmYear{2021}
\copyrightyear{2021}
\acmSubmissionID{fse21main-p1016-p}
\acmISBN{978-1-4503-8562-6/21/08}
\acmConference[ESEC/FSE '21]{Proceedings of the 29th ACM Joint European Software Engineering Conference and Symposium on the Foundations of Software Engineering}{August 23--28, 2021}{Athens, Greece}
\acmBooktitle{Proceedings of the 29th ACM Joint European Software Engineering Conference and Symposium on the Foundations of Software Engineering (ESEC/FSE '21), August 23--28, 2021, Athens, Greece}

\begin{document}
\title{\tool: Guaranteeing Generalization in Programming by Example}

\author{Bo Wang}
\affiliation{%
  \institution{National University of Singapore}
  \country{Singapore}
}
\email{bo_wang@u.nus.edu}

\author{Teodora Baluta}
\affiliation{%
  \institution{National University of Singapore}
  \country{Singapore}
}
\email{teodora.baluta@u.nus.edu}

\author{Aashish Kolluri}
\affiliation{%
  \institution{National University of Singapore}
  \country{Singapore}
}
\email{e0321280@u.nus.edu}

\author{Prateek Saxena}
\affiliation{%
  \institution{National University of Singapore}
  \country{Singapore}
}
\email{prateeks@comp.nus.edu.sg}

\begin{abstract}
    Programming by Example (PBE) is a program synthesis paradigm in which the synthesizer creates a program that matches a set of given examples. In many applications of such synthesis (e.g., program repair or reverse engineering), we are to reconstruct a program that is close to a specific target program, not merely to produce some program that satisfies the seen examples. 
    In such settings, we wish that the synthesized program {\em generalizes} well, i.e., has as few errors as possible on the unobserved examples capturing the target function behavior. In this paper, we propose the first framework (called \tool) for PBE synthesizers that guarantees to achieve low generalization error with high probability. Our main contribution is a procedure to dynamically calculate how many additional examples suffice to theoretically guarantee generalization. We show how our techniques can be used in 2 well-known synthesis approaches: PROSE and STUN (synthesis through unification), for common string-manipulation program benchmarks. We find that often a few hundred examples suffice to provably bound generalization error below $5\%$ with high ($\geq 98\%$) probability on these benchmarks. Further, we confirm this empirically: \tool significantly improves the accuracy of existing synthesizers in generating the right target programs. But with fewer examples chosen arbitrarily, the same baseline synthesizers (without \tool) {\em overfit} and lose accuracy.
\end{abstract}

%------------------------

%------------------------


\maketitle              % typeset the header of the contribution

\section{Introduction}
\label{sec:intro}

Program synthesis is the goal of automatically generating computer programs for a given task. This vision has existed for over at least four decades~\cite{sammet1966EnglishProgram,smith1975pygmalion,summers1977methodology}. One of the mainstream approaches towards this goal is programming by example (or PBE)~\cite{gulwani2016programming}. In its simplest form, a PBE synthesizer is given access to an {\em oracle} that can generate correct input-output (I/O) examples for the unknown {\em target} program. The synthesizer has to create a candidate program as close as possible to the target program from a pre-specified {\em hypothesis space}, i.e., the space of all possible candidate programs that the synthesizer can reason about. The number of given I/O examples can vary depending on the end application, but the fewer the better. Therein lies the challenge of {\em generalization}: If examples are too few, then many possible candidate functions satisfy them, and picking one arbitrarily might yield a solution that works well only on the seen examples. In other words, the solution overfits to the seen examples and may not generalize well. How does a synthesizer create programs that are provably close to the target program? This has been a fundamental question for PBE-based program synthesis.

There are several domain-specific solutions to generalization. In program repair
as well as in inductive synthesis, for example, inferring additional specifications from observed examples that must be satisfied by the program is shown to help with generalization~\cite{laich2020guiding,jha2010oracle,goues2019apr,an2019augmented}. Allowing the synthesizer to use
more powerful oracles that adaptively craft examples or logical invariants help to synthesize correct programs~\cite{ezudheen2018ICE,ji2020question}. In neural-guided program synthesis~\cite{shin2018improving, chen2018execution, devlin2017robustfill, polosukhin2018neural}, machine learning techniques to avoid over-fitting such as regularization or structural risk minimization, are employed implicitly~\cite{vapnik2000nature}. 
In domains where we have prior knowledge about the likely distribution to which the target program belongs, synthesizers can rank solutions~\cite{singh2015predicting}, guide program search~\cite{laich2020guiding}, and use generative models for program representation or distributional priors~\cite{ellis2017learning}. 
Some synthesizers favor short programs as per Occam's razor~\cite{peter2007mdl}.

All of the above approaches, while useful, require additional knowledge or implicit assumptions about the target program beyond that captured by the original PBE problem setup. Without such assumptions, these approaches do {\em not} provide any formal guarantee that the produced program will be correct or generalize well on unseen examples. It is natural to ask: Can we guarantee generalization without making any additional domain-specific assumptions?

In this paper, we study generalization in PBE from the perspective of {\em sample complexity}: How many I/O examples should the synthesizer have to see to be confident that its selected solution is close to the target program? To answer this question, the PAC learning theory provides a starting point~\cite{valiant1984PAC,blumer1987occam}. A synthesizer generalizes well when the synthesized program is close to the target program with high confidence. The notion of confidence and closeness to the target program can be made formal using PAC learning theory.
Specifically, the synthesized program generalizes if it will make no more than a small fraction $\epsilon$ of errors on unseen examples taken independently, with high probability (at least $1-\delta$).

Our approach works on any distribution that the I/O examples are sampled from. To formally guarantee that generalization is achieved on the distribution, we need a principled design for PBE synthesizers. Existing PBE synthesizers are not designed to provide generalization guarantees; therefore, they pick the number of I/O examples to work with in an ad-hoc fashion. For instance, they may synthesize a program after seeing only $2-4$ examples~\cite{singh2015predicting,gulwani2011automating}. This paper seeks to answer the following questions: 

\paragraph{RQ1.} How many I/O examples would a synthesizer need to see in order to provably generalize?

\paragraph{RQ2.} Do existing synthesizers overfit with, say, $2-4$ examples?

As a conceptual contribution, we present the first principled framework, \tool\footnote{The tool is available with DOI number 10.5281/zenodo.4883273. The latest version can be found at \url{https://github.com/HALOCORE/SynGuar}}, to provide generalization guarantees about the synthesized programs. We propose a procedure that computes the size of the hypothesis space {\em dynamically} during synthesis, which is then used to calculate the sample size required to provably generalize.
The challenge is therefore two-fold. First, while efficiently computing the size of the hypothesis space is easy in some existing PBE synthesizers (such as the PROSE framework~\cite{polozov2015flashmeta}), for others it requires careful design. An example of the latter is the synthesis through unification (STUN)~\cite{Alur2015stun} approach.
As our main technical contribution, we present an example of integrating \tool into synthesizers based on PROSE and STUN frameworks. Specifically, we provide two PBE synthesizers for string manipulation programs that provably generalize, one implemented in the PROSE framework (\toolvsa) and one based on the STUN approach (\toolstun). To the best of our knowledge, no prior PBE synthesizer claims such strong generalization guarantees about the synthesized programs.

From an empirical perspective, our work provides the first experimental evidence for the number of examples sufficient to guarantee generalization in practice for simple string-manipulation tasks. We run \toolvsa and \toolstun on two benchmarks in this domain: 1) manually designed data-wrangling tasks similar to those used in FlashFill~\cite{polozov2015flashmeta}; and 2) the standard \sygus 2019 benchmark~\cite{sygus-benchmark}, respectively. We find that on their respective benchmarks, the tools produce programs which are provably within at most $5\%$ generalization error with a modest number of examples---around $197$ samples and $357$ examples on average, respectively---with a high probability ($\geq 98\%$). 
This observation also suggests that it is unlikely for PBE synthesizers to generalize from just $2-4$ examples without using some implicit or explicit additional knowledge. We confirm this observation by running the vanilla versions, i.e., versions with \tool disabled, on the same benchmarks with 4 randomly chosen examples or the given seed examples in the benchmark ($2-10$ in size). We find that \toolvsa generates the correct target program for $14/16$ cases from the data-wrangling task benchmark and \toolstun for $53/59$ cases from the \sygus benchmark. In contrast, the vanilla versions generate correct programs for $0/16$ and $36/59$ cases, respectively.
%These findings form our main empirical contributions.
This shows that without enough examples, synthesized programs often overfit.

Though we focus on string-manipulating programs, our approach makes minimal additional assumptions and thus can be extended to other  application domains, such as program repair~\cite{goues2019apr}, invariant discovery~\cite{blazytko2020Aurora} and so on. 
The generalization guarantee fits well in applications such as data cleaning and transformation \cite{dasu2003exploratory}, where a provable accuracy matters, or automatic stub writing in symbolic execution \cite{mechtaev2018symbolic,Shen2019NeuroSymbolicEA}, where the goal is to learn a symbolic constraint that is \emph{approximately} close enough to the target.
Our experiments suggest that the sample size to achieve generalization is task and benchmark dependent, which leaves the question of how well our presented approach works in other domains or benchmarks open. These are grounds for promising future work.

\section{Overview}
\label{sec:problem}

\begin{figure}
\centering
\begin{lstlisting}[linewidth=0.95\linewidth,
xleftmargin=7pt,
frame=single,
basicstyle=\small\ttfamily] 
language StrPROSE; 
@input string x;
@start string program := recTerm;
string recTerm := catTerm | Concat(catTerm, recTerm);
string catTerm := ConstStr(cs) | convTerm;
string convTerm := term | UpperCase(term) | LowerCase(term);
string term := SubStr(x, pos1, pos2);
int? pos := AbsPos(x, ka) | RegPos(x, kr);
string cs;  //constant string
int ka;     //absolute position
Tuple<Regex,Regex,int,int> kr; // kr=(r1,r2,k,offset)
\end{lstlisting}
\caption{A DSL for string transformation programs. \texttt{Concat} returns the string produced by concatenating two strings \texttt{catTerm} and \texttt{recTerm}, \texttt{SubStr} returns the substring between \texttt{pos1} and \texttt{pos2}.
The \texttt{UpperCase} and \texttt{LowerCase} return the string in upper case and lower case, respectively. \texttt{AbsPos} returns the absolute position of string \texttt{x}. The \texttt{RegPos} operator outputs the \texttt{k}$^{th}$ position plus an \texttt{offset} where the boundaries of the strings returned by applying regular expressions \texttt{r1} and \texttt{r2}, respectively, on string \texttt{x} match.}
\label{fig:dsl}
\end{figure}

\begin{figure*}
\begin{minipage}[b]{.48\linewidth}
\centering
%// Target program
%userdefinedwidth=0.46\textwidth,
\begin{mdframed}[
innerbottommargin=0pt,
innertopmargin=0pt,
innerleftmargin=2pt,
innerrightmargin=2pt,
%linewidth=1pt,
]
%linewidth=0.46\textwidth,
%xleftmargin=5pt,
% frame=single,
% rulecolor=\color{black},
% lineskip=-1ex,
\begin{lstlisting}[
frame=none,
lineskip=-1.6pt,
basicstyle=\small\ttfamily,
keywordstyle=\color{black}\bfseries,
morekeywords={Concat, SubStr, RegPos,AbsPos, ConstStr},
commentstyle=\color{blue}
]
// word \w+=(A-Za-z0-9)+, digit \d+=(0-9)+
// *.?= any character
Concat( // return substr until the end of 1st word match
 SubStr(x, AbsPos(x, 0), RegPos(x, (\w+, .*?, 0, 0))),
 Concat(ConstStr(","), // append ","
  Concat( // return substr of 2nd word match until the end of the word
   SubStr(x, RegPos(x, (.*?, \w+, 1, 0)),
              RegPos(x, (\w+, .*?, 1, 0))),
   Concat(ConstStr(","),
    Concat( // return the substr of 3rd word match until the end of the word
     SubStr(x, RegPos(x, (.*?, \w+, 2, 0)),
                RegPos(x, (\w+, .*?, 2, 0))),
     Concat(ConstStr(","),
      // return the substr from the 4th word until the end of the string
       SubStr(x, RegPos(x, (.*?, \w+, 3, 0)),
                 AbsPos(x, -1))))))))
\end{lstlisting}
\end{mdframed}
\end{minipage}%
\hfill
\begin{minipage}[b]{.5 \linewidth}
\centering
  %// Overfitted synthesized program
  %userdefinedwidth=0.46\textwidth,
\begin{mdframed}[
innerbottommargin=0pt,
innertopmargin=0pt,
innerleftmargin=2pt,
innerrightmargin=2pt,
%linewidth=1pt,
]
%linewidth=0.5\textwidth,
% frame=single,
% rulecolor=\color{black},
\begin{lstlisting}[
frame=none,
lineskip=-3pt,
basicstyle=\small\ttfamily,
keywordstyle=\color{black}\bfseries,
morekeywords={Concat, SubStr, RegPos,AbsPos, ConstStr},
commentstyle=\color{blue},
escapechar=!]
Concat(  // return substr of first two characters
 !\lcolorbox{yellow}{SubStr(x, AbsPos(x, 0), AbsPos(x, 2)),}! !\textcolor{red}{overfits}!
 Concat(ConstStr(","),
  Concat( // return the string from the start of the first number offset by 1 till the second last separator
   !\lcolorbox{yellow}{SubStr(x, RegPos(x, (.*?, (-?\textbackslash d+)(\textbackslash.\textbackslash d+)?, 0, 1)),}!
              !\lcolorbox{yellow}{RegPos(x, (.*?, [\textbackslash,\textbackslash.\textbackslash;\textbackslash-\textbackslash|], -2, 0))),}! !\textcolor{red}{overfits}!
   Concat(ConstStr(","), 
    Concat( // return the substr from second capital letter offset +2 to offset +3
     !\lcolorbox{yellow}{SubStr(x, RegPos(x, (.*?, [A-Z]+, 1, 2)),}!
                !\lcolorbox{yellow}{RegPos(x, (.*?, [A-Z]+, 1, 3))),}! !\textcolor{red}{overfits}!
     Concat(ConstStr(","),
      // return the substr from start of second word with offset +3 till the end
       !\lcolorbox{yellow}{SubStr(x, RegPos(x, (\textbackslash w+, .*?, 1, 3)),}!
                   !\lcolorbox{yellow}{AbsPos(x, -1))))))))}! !\textcolor{red}{overfits}!
\end{lstlisting}
\end{mdframed}
  \end{minipage}
\caption{%
On the left, we show the correct target program $t$ that tokenizes the string \texttt{x} on the boundaries of the first 4 words.
On the right, we show the program synthesized on $2$ examples which overfits at all the four \texttt{SubStr} operations.}
\label{fig:example}
\end{figure*}

\begin{figure}[ht]
    \centering
\begin{tabular*}{0.45\textwidth}[b]{c|c|c|c}
\textbf{\#Ex.} & \textbf{Input} & \textbf{Output} & \textbf{\#$f$} \\ \hline
1 & \verb+"0E-E 2|u7kuZ85"+ & \verb+"0E,E,2,u7kuZ85"+ & $\approx 10^{42}$ \\
2 & \verb+" J-3bJ.9;PPm"+ & \verb+" J,3bJ,9,PPm"+ & $\approx 10^{20}$ \\
3 & \verb+"tpJ|AV n0d7 6z"+ & \verb+"tpJ,AV,n0d7,6z"+ & $\approx 10^{12}$ \\
4 & \verb+"R 3|6VCs Q"+ & \verb+"R,3,6VCs,Q"+ & $304128$ \\
5 & \verb+"M x cSkrw ru6"+ & \verb+"M,x,cSkrw,ru6"+ & $304128$ \\
6 & \verb+"Wk  U U nZp X "+ & \verb+"Wk,U,U,nZp X "+ & $864$ \\
7 & \verb+"gsa-ub  hn lpa"+ & \verb+"gsa,ub,hn,lpa"+ & $216$ \\
8 & \verb+"RO8I3  R SuM|e "+ & \verb+"RO8I3,R,SuM,e "+ & $144$ \\
9 & \verb+"q E 0 LD0 "+ & \verb+"q,E,0,LD0 "+ & $144$ \\
10 & \verb+"dZPz T.Q s "+ & \verb+"dZPz,T,Q,s "+ & $144$ \\
11 & \verb+"Ny e87e -lO  0w"+ & \verb+"Ny,e87e,lO,0w"+ & $36$ \\
12 & \verb+"FX 1 U P 1fN"+ & \verb+"FX,1,U,P 1fN"+ & $18$ \\
%13 & Ti iiK4s 3 F7v9 & Ti,iiK4s,3,F7v9 & $18$
\end{tabular*}%
\caption{The I/O examples provided to the synthesis algorithm, along with the number of consistent candidates (\#$f$). This number drops from $\approx 10^{42}$ to $18$ in just $12$ samples.
}
\label{fig:example1}
\end{figure}

In practice, it is hard to know how many I/O examples suffice to solve a synthesis task. The number of I/O examples across various target programs, even for the same synthesizer, vary in prior works and are chosen somewhat arbitrarily. For instance, the \sygus benchmark~\cite{sygus-benchmark} has a dedicated track for the domain of string-manipulating programs. The benchmark consists of several PBE tasks and each of them is provided with a different number of I/O examples. Some have as low as $2$ examples whereas others have $50$. 

To illustrate the problem of overfitting in PBE-based synthesis, let us consider a data-wrangling task of tokenizing a given passage of text into individual words. The tokenization task involves recognizing 
the boundaries of the first 4 words and replacing the characters used to separate them with commas.
The user wishes to use a program synthesis tool to learn a program that performs this task. Here, we use our tool \toolvsa which is based on the PROSE framework~\cite{polozov2015flashmeta} for synthesizing a program. Program synthesizers
output programs in the syntax of some pre-specified {\em target language} or domain-specific language (DSL). In order to support our task, we implement a string transformation DSL in PROSE that is similar to the FlashFill DSL~\cite{polozov2015flashmeta}, see Figure~\ref{fig:dsl}. 
We give a reference implementation of our modified DSL in the supplementary material~\cite{aaasup}.
The user provides an oracle, which can be queried by the synthesizer for I/O examples exemplifying the behavior of the target program. Each I/O example is a pair of strings formed from lower and upper case letters,  digits, spaces and separators.

Given the problem setup as described above, the goal of the synthesizer reduces to learning the 4 correct substrings from the provided examples. Let us examine how well the synthesizer performs on a few I/O examples, say $2$. After running it on the first $2$ examples given in Figure~\ref{fig:example1}, as expected, the output program overfits. For instance, instead of trying to get the first substring until the end of the first word, it just picks two characters for the first word since both the examples have only two characters until their respective first separators. In fact, the synthesizer continues to overfit even after being provided $9$ examples. 

However, it turns out that for our running example, we can theoretically assert a program close to the target will be picked with high probability after $149$ examples! Our proposed algorithmic framework \tool is able to calculate this quantity on the fly and stop when enough examples are seen. To understand how it works, consider Figure~\ref{fig:example1} that shows the estimated size\footnote{A sound upper bound of the actual size is calculated.} of the hypothesis space that is {\em consistent} (or matches) with all the seen examples up to a certain point. \tool computes this quantity internally, which serves as our main technical insight. Notice that the space of the consistent programs shrinks progressively as more examples are provided. Before seeing any example, the hypothesis space is the set of all the programs our target language can represent and its size can be infinite as the DSL grammar is recursive. Now suppose the synthesizer sees the first I/O example, then the number of candidate programs which are {\em consistent} with the first example reduces considerably. The reduction depends critically on the example provided. For instance, the first I/O example shown in Figure~\ref{fig:example1} will reduce the space of consistent programs to $10^{42}$. This is still quite large---if we arbitrarily pick one program, without using any auxiliary assumptions or prior knowledge about the target program, the odds of picking the correct program are negligibly small.
However, after seeing $12$ examples, the consistent program space reduces to $18$ for our running example. 
Note that choosing a program out of these $18$ at random does not have any guarantees on its closeness to the target program. 
For a program space of $18$, to choose a program that is close to the target program with provably high confidence we require $137$ more examples. 
\tool can provably assert that it has seen enough examples to stop and return a program close to the target program after $12+137=149$ examples.

\paragraph{Problem Setup.} Similar to the setup used in existing PBE-based synthesizers, we are given an oracle to query I/O examples and a DSL for representing
the output program. Additionally, we are given  user-specified $(\epsilon,\delta)$ parameters that capture the desired generalization guarantee. The synthesizer queries the oracle for as many I/O examples as it needs and terminates with either $None$ or a synthesized function $f$. The probability that the synthesizer returns a function $f$ that might not generalize should be under the given small $\delta$.
Let $S=\{(x,~t(x))\}$ be the set of I/O examples seen by one invocation of the synthesizer. Here, each $x$ is an input drawn independently and identically distributed (i.i.d) from an unknown distribution $\distrib$ that the oracle captures.
We assume that $f$ will satisfy all given I/O examples, $\forall x\in S, f(x)=t(x)$. Note that this is different from the ``best-effort''~\cite{peleg2020perfect} or approximate synthesis approaches~\cite{Shen2019NeuroSymbolicEA} where the program $f$ is allowed to differ from $t$ on some examples in $S$. In this setup, therefore, $S$ is a random variable. The synthesizer, denoted as $\synthalgo$, is also a random variable defined over I/O samples ($S$) drawn from $\distrib$.

We seek to design PBE synthesizers that achieve generalization given by a rigorous PAC-style guarantee \cite{valiant1984PAC} while computing the required sample complexity. For an $(\epsilon,\delta)$-synthesizer, the generalization error $error(f) = Pr_{x \sim \distrib} \lbrack f(x) \neq t(x) \rbrack$ is bounded by $\epsilon$. The probability of generating $f$ with $error(f) > \epsilon$ is bounded by $\delta$.

\begin{definition}[$(\epsilon,\delta)$-synthesizer]
\label{def:eps-delta}
  A synthesis algorithm $\mathcal{A}$ with hypothesis space $H$ is an $(\epsilon, \delta)$-synthesizer with respect to a target class of
  functions $\concept$ iff for any input distributions $\distrib$, for all $t
  \in \concept$, $\epsilon \in (0, 1)$, $\delta \in (0, 1)$, given example set $S$ drawn i.i.d from the $\distrib$,
  \begin{align*}
    Pr\lbrack \synthalgo~\text{outputs}~ f \in H \text{ such that } error(f) > \epsilon \rbrack < \delta
  \end{align*} 
\end{definition}

\section{The \tool Framework}
\label{sec:overview}

\begin{algorithm}[t]
    \caption{Meta-synthesis Algorithm}
    \label{alg:meta-algo}
    \begin{algorithmic}[1]
      \Procedure{MetaSyn}{$\epsilon, \delta$}
    \While {\colorbox{lime}{stopping condition}}%\begin{tcolorbox}[width=2in,height=0.2in,colframe=white,colback=yellow]stopping condition\end{tcolorbox}}
      \State Query user for k examples
      \State \colorbox{lime}{Update the hypothesis space}
      \State \Return $None$ if empty hypothesis space
      \EndWhile
      \State Compute the $m$ examples for $(\epsilon, \delta)$ guarantee
      \State Query user for $m$ examples
      \State \colorbox{lime}{Update the hypothesis space}
      \If{empty hypothesis space}
        \State \Return $None$
      \EndIf
      \State \Return $f$ in hypothesis space
      \EndProcedure
    \end{algorithmic}
\end{algorithm}

We propose a framework with a similar algorithmic meta-structure as that of existing PBE engines.
The overall procedure is shown in Algorithm~\ref{alg:meta-algo}.
Instead of synthesizing a program after seeing a pre-determined number of I/O examples the procedure queries an oracle for new examples until it synthesizes a program that generalizes. There are two key new features in our framework: a stopping criterion and a dynamically calculated count of the number of samples to be seen. The following classical result gives us a starting point to compute the count precisely.

\paragraph{A Starting Point.} The number of examples provably sufficient to achieve the $(\epsilon, \delta)$-generalization is given by Blumer \etal~\cite{blumer1987occam}. We restate this result, which computes sample complexity as a function of $(\epsilon,\delta)$ and the capacity (or size) of any given hypothesis space $H$.

\begin{theorem}[Sample Complexity for $(\epsilon, \delta)$-synthesis]
  \label{thm:static-union-bound}
For all $\epsilon \in (0,1)$, $\delta\in
(0,1)$, and hypothesis space $H$, a synthesis algorithm $\synthalgo$ which outputs functions consistent with $m$ i.i.d samples is an $(\epsilon, \delta)$-synthesizer, if 
\begin{align*}
  m > \frac{1}{\epsilon}(\ln |H| + \ln
  \frac{1}{\delta})\\
\end{align*}
\end{theorem}

The above theorem is intuitively based on the following analysis. Let us say we have some initial hypothesis space $H$. After seeing one new I/O example, each hypothesis that is ``$\epsilon$-far'' from $t$ (generalization error $>\epsilon$) becomes inconsistent with some non-zero probability, and is eliminated. Therefore after seeing sufficiently many new examples, the probability of any ``$\epsilon$-far'' hypothesis being output falls below $\delta$. For details, please see the analysis \cite{blumer1987occam}.

\subsection{Key Observations \& Challenges}

We observe that Theorem~\ref{thm:static-union-bound} can be used at any point of the synthesis procedure. After seeing say the first $S$ examples, let the space of programs consistent with the examples be $\programsp$. We can plug $|\programsp|$ into Theorem~\ref{thm:static-union-bound} to compute how many more examples are sufficient to achieve the guarantee provided in Definition~\ref{def:eps-delta}.
But, there are several key technical challenges in utilizing the classical result of Theorem~\ref{thm:static-union-bound} in providing end-to-end generalization guarantees. 

First, applying this result requires being able to compute $|\programsp|$. We point out that this has not been an explicit algorithmic goal when designing existing synthesizers. Consequently, computing $|\programsp|$ is non-trivial in some of the existing synthesizers.
To tackle this, we design our own PBE synthesizer based on the STUN approach and bottom-up explicit search with the ability to compute $|\programsp|$ (see Section~\ref{sec:synguar-stun}).
Further, we show how to integrate \tool in PROSE, a synthesis framework where the size of the hypothesis space can be easily computed.

Second, $|\programsp|$ can be large and plugging in its values at the beginning leads to vacuously high sample bounds in practice. For instance, initially the hypothesis size in our running task is infinite, and even after seeing one example, the size is $10^{42}$. Therefore, instead of naively plugging in values of parameters at the beginning, we use the idea that if $|\programsp|$ decreases as the synthesizer sees more examples then the estimated sample complexity for generalization reduces as well. Therefore, in \tool's design, $|\programsp|$ is computed on the fly as the synthesizer sees more examples and the stopping condition ensures that the synthesized program generalizes.

Lastly, the PAC learning theory offers no recourse to predict how fast $|\programsp|$ reduces with more samples in practice. This question, then, becomes an empirical one: \textit{For which programs do we observe that a small number of examples are sufficient to generalize?}
Our empirical evaluation shows that for several common string manipulation tasks, the required number of examples turn out to be modest.

\paragraph{Remark.} It can be seen that our solution is a modification to the existing PBE synthesis loop, which can be instantiated for several program synthesis engines. Our proofs and analysis utilize classical sample complexity arguments, together with bounds for hypothesis space size. It may therefore appear surprising then that such calculations are not routine in prior program synthesis works already, despite accuracy / generalization being a natural objective. We believe that the challenges stated above offer an explanation as to why applying previous theoretical results is not straightforward, and requires a principled approach. We point out that the subtlety is in our problem formulation itself, namely, the use of dynamic calculation of the remaining sample size needed for generalization.

\begin{algorithm}[t]
    \caption{\tool Synthesis returns a program with error smaller than $\epsilon$ with probability higher than $1-\delta$}
    \label{alg:main}
    \begin{algorithmic}[1]
      \Procedure{\tool}{$\epsilon, \delta$}
      \State $k = 1$ // tunable parameter
      \State $g \gets \Call{PickStoppingCond}{}$ \label{alg:pick-cond}
      \State $S' \gets \varnothing, s \gets 0$
      \State $size_{H} \gets \Call{ComputeSize}{H}$ \label{alg:call-compute-init} 
      \State $n \gets g(size_{H})$
      %\While {$|S'| \leq n$}
      \While {$s \leq n$}
      \State $S' \gets S' \cup \Call{sample}{k}$
      \State $H_{S'} \gets $ \Call{UpdateHypothesis}{$S'$} \label{alg:call-update}
      %\State $H_{S'}$ = synthesize($S'$)
      \State $size_{H_{S'}} \gets \Call{ComputeSize}{H_{S'}}$ \label{alg:call-compute}
      \State \Return $None$ if $size_{H_{S'}} = 0$
      \State $s \gets s + k$
      \State $n \gets min(n, s+g(size_{H_{S'}}))$ \label{alg:min-thresh}
      \EndWhile
      \State $m_{H_{S'}}=\frac{1}{\epsilon}(\ln{size_{H_{S'}}} + \ln{\frac{1}{\delta}})$  \label{alg:validate-samplesize}
      \State $T \gets \Call{sample}{m_{H_{S'}}(\epsilon, \delta)}$
      \State $S \gets S'\cup T$ 
      \State \Return program $f$ in $H_S$ or None if  $H_S = \varnothing $
      \EndProcedure
    \end{algorithmic}
\end{algorithm}

\subsection{\tool Algorithm}
\label{sec:approach}

We start by addressing the second challenge assuming that $|\programsp|$ is computable.
Recall that, our synthesis algorithm takes as input the error tolerance parameters
$\epsilon$ and the confidence $\delta$ (Algorithm~\ref{alg:main}).
The algorithm follows the structure of Algorithm~\ref{alg:meta-algo} and consists of two phases: the sampling phase (lines $7-13$) and the validation phase (lines $14-17$).
The sampling phase addresses the second challenge by trying to shrink the hypothesis space as much as possible. In each iteration $k$ samples are taken before updating the hypothesis space that satisfies all $k$ samples seen.
The crucial part of the sampling phase is deciding when the number of examples seen so far (stored in the set $S'$ and whose cardinality is $s$) is ``enough''.
\tool stops sampling when the number of examples seen so far exceeds a stopping threshold, represented by variable $n$.
In each iteration, the stopping threshold (which depends on finite $|H|$ initially) either remains the same or it gradually shrinks with each update of the hypothesis space (line~\ref{alg:min-thresh}).
Hence, \tool ~{\em dynamically} updates the threshold based on the change in the hypothesis size.
To control how much the threshold variable shrinks with respect to the size of the updated hypothesis space $H_{S'}$, \tool picks a function $g$ (line~\ref{alg:pick-cond}). 
%While one can choose any function $g:\mathbb{N}\mapsto\mathbb{Z}$,
Our framework allows using any $g:\mathbb{N}\mapsto\mathbb{Z}$ that is monotonically non-decreasing, and we provide a sample complexity analysis for such functions. We propose a particular choice of function $g$ as the default: 
$\parf(x)=\max\{0, ~\frac{1}{\epsilon}(\ln(x)-\ln(\frac{1}{\delta}))\}$.
This $g$ has a useful property---the required number of samples it entails in the worst case cannot be more than twice the number of samples that an optimal choice of $g$ will take.
%{\em optimal} number of samples.
We will formally state and prove this optimality claim in Section~\ref{sec:analysis}.

In the second phase, in addition to the samples $S'$, \tool samples a fixed number of samples according to Theorem~\ref{thm:static-union-bound}. \tool then can return a program $f$ with provable $(\epsilon, \delta)$ guarantees (line~\ref{alg:validate-samplesize}).
Algorithm~\ref{alg:main} calls sub-procedures \textsc{UpdateHypothesis} (line~\ref{alg:call-update})
to find a program space consistent with $S'$ and \textsc{ComputeSize} (line~\ref{alg:call-compute})
to compute the size of the consistent program space. The sub-procedure \textsc{UpdateHypothesis} can be implemented by any existing PBE synthesis algorithm
which return hypotheses consistent with $S'$. Section~\ref{sec:design-synth} details how to implement \textsc{ComputeSize}, which is specific to the underlying \textsc{UpdateHypothesis} sub-procedure.

\paragraph{Running Example. } Consider the example given in Section~\ref{sec:problem}. First the user inputs $\epsilon=5\%, \delta=2\%$ respectively. In the sampling phase,
the user is queried for one example in each iteration. After the first iteration, i.e., seeing one example, the sample size for generalization ($m_{H_{S'}}$) is $2018$. Instead, \tool's sampling phase stops after seeing $n=12$ examples and the additional sample size for generalization (Line \ref{alg:validate-samplesize}) reduces to $m_{H_{S'}}=137$. The total sample size of both phases sums up to $149$ examples, which is $10\times$ less than the sample complexity after the first iteration. This is a direct consequence of \tool dynamically estimating the sample size for generalization.

\subsection{Analysis of the Algorithm}
\label{sec:analysis}

\tool's design is motivated by being able to give a formal generalization guarantee and a bounded sample complexity. For this purpose, we state and prove the following properties: 
% Here, we detail and prove these properties.

\paragraph{(P1: Termination)} \tool always terminates for a finite $|H|$.

\paragraph{(P2: $(\epsilon, \delta)$ guarantees)} The probability of \tool returning an $f$ that is $\epsilon$-far is smaller than $\delta$.

\paragraph{(P3: Sample complexity)} \tool's sample complexity is always within 2$\times$ of the optimal for $k\leq \frac{1}{2\epsilon}\ln\frac{1}{\delta}$.

\begin{theorem}[P1]
\label{thm:p1}
\tool always terminates for a finite $|H|$.
\end{theorem}
\begin{proof}
It suffices to prove that the sampling phase (lines $7-13$) of \tool terminates in order to show that \tool terminates.
%The phase can be concisely written as follows:
In each iteration of the sampling phase,  let $S_i$ be the queue storing the user-provided examples after each iteration, $z_t$ be the $t^\text{th}$ example, $S_{i+1} = S_{i}\cup \{z_{ik+1},...z_{ik+k}\}$ and $S_0=\varnothing$. For each $S_i$, $H_{S_i}$ determines the set of consistent hypothesis that satisfy $S_i$. Let $N_i$ be the limit of  the number of I/O examples $n$ for the sampling phase after iteration $i$. 
For iterations $i$ and $j$ where $i<j$ and $\forall g:\mathbb{N} \rightarrow \mathbb{Z}$ such that $g$ is monotonically non-decreasing, the following holds: 
\begin{align*}
  &S_i\subset S_j \Rightarrow |H_{S_j}|\leq |H_{S_i}| \Rightarrow g(|H_{S_j}|) \leq g(|H_{S_i}|)\\
  &N_j \leq \min\{N_i, |S_j| + g(|H_{S_j}|)\} \leq N_i \text{ (see line $13$ in Alg.~\ref{alg:main})}
\end{align*}
Therefore, if $N_0 \leq g(|H|)$ then the loop will terminate at some iteration $p$ such that $N_p < |S_p| \leq N_p + k \leq N_0 + k$.
\end{proof}

\begin{theorem}[P2]
\label{thm:p2}
The probability of \tool returning a synthesized program $f$ that is $\epsilon$-far is smaller than $\delta$.
\end{theorem}
\begin{proof}
By Theorem~\ref{thm:p1}, we know that the sampling phase terminates with $S'$ samples (see line $14$). In lines $14-16$ \tool samples an additional number of I/O examples required to generalize and then synthesizes a program after seeing the additional samples. Therefore, Theorem~\ref{thm:p2} follows from Theorem~\ref{thm:static-union-bound}.
\end{proof}

In order to prove the last property, we define a new quantity $\omega(Q)$. It is the smallest sample size taken by \tool($\epsilon, \delta$) for any non-decreasing $g$ used for a sequence of I/O examples $Q$.

\begin{definition}[Smallest dynamic sample size]
    For any infinite sampled sequence of examples $Q$, let $\Call{Prefix}{Q,g}$ be the prefix of $Q$ at which \tool($\epsilon, \delta$) terminates. Then,  
    % In other word, it is like that the best  stopping condition that result in smallest sample size is known before seeing the examples.
    $$
    \omega(Q) = \inf \{ |m_g|~:~\forall g,\text{ } m_g = \Call{Prefix}{Q,g}\}
    $$
\end{definition}

\begin{theorem}[P3]
\label{thm:p3}
    \tool uses no more than $2\omega(Q)$ examples on any $Q$ when the result is not None with $\parf(x)=\max\{0, ~\frac{1}{\epsilon}(\ln(x)-\ln(\frac{1}{\delta}))\}$ and $k\leq \frac{1}{2\epsilon}\ln\frac{1}{\delta}$.
\end{theorem}
Due to space limit, we provide the proof of this theorem in the supplementary material~\cite{aaasup}.

\section{Retrofitting \tool into Existing Synthesizers}
\label{sec:design-synth}

We now show how to compute $|\programsp|$, the size of the consistent program space.
A sound upper bound of $|\programsp|$ is safe to use, since in this case, our analysis shows \tool to take more examples than those needed to guarantee generalization.
We show how to compute $|\programsp|$ bounds for two well-known PBE synthesis approaches.

\subsection{\tool for the PROSE Framework}
\label{sec:synguar-prose}

We first apply \tool on top of the PROSE framework~\cite{polozov2015flashmeta}, a state-of-the-art PBE meta-synthesis framework that generates an inductive synthesizer for a given DSL. PROSE allows developers to write DSLs and specify {\em witness functions} that capture (a subset of) inverse semantics for the DSL operators. These witness functions are the drivers for the ``deductive backpropagation'' because they specify the inputs or properties of the input given an I/O example.

We implement a synthesizer named \gprose with the DSL in Figure~\ref{fig:dsl} on top of PROSE by specifying executable semantics and witness functions for its operators. Our DSL shares most operators with the DSL of FlashFill~\cite{polozov2015flashmeta}.
For the operators that differ, we detail their executable semantics and witness functions in the supplementary material~\cite{aaasup}.
Note that \tool works with any DSL expressible in PROSE, as long as each operator has its semantics and an associated witness function specified.

PROSE uses an internal succinct representation of the program space using a data structure called version-space algebra (VSA) which makes it convenient to calculate $|\programsp|$~\cite{mitchell1982generalization,lau2000version,lau2003programming,gulwani2011automating}.
A VSA is a directed graph where each node corresponds to a set of programs. The leaf nodes explicitly represent a set of programs that can be enumerated. There are two types of internal nodes: {\em union} nodes that represent a set-theoretic union and {\em join} nodes that represent $k$-ary operators which are defined by the DSL. 

\paragraph{Computing $|\programsp|$ Using VSA.}
PROSE readily computes $|\programsp|$ using a bottom-up graph traversal on its VSA. For each leaf node, it enumerates and counts the set of programs directly. For every union node, to compute the corresponding number of programs it adds up the count of all child nodes. For every join node, the number of programs is a cross product of all applications of the $k$-ary operator to $k$ parameter programs. This soundly upper bounds
$|\programsp|$.
In our implementation, we reuse the \textsf{Size} API available in PROSE, resulting in the sizes shown in Figure~\ref{fig:example}.

\paragraph{Scaling to Large Sample Size.}
Building VSA on a large number of examples can be time-consuming. Therefore, we build the VSA on a subset of the examples which lead to the same set of programs. More specifically, we take the examples one by one and drop the examples that do not decrease the VSA size.

\subsection{\tool in \gstun}
\label{sec:synguar-stun}

\stun is a well-known synthesis approach~\cite{Alur2015stun}. It was originally proposed as an extension of  the counter-example guided inductive synthesis (CEGIS) approach to synthesize program from the specification. The high-level idea is to synthesize partial solutions satisfying  
parts of the inputs and unify them.  As an instantiation of STUN for synthesizing conditional programs under PBE settings, we work with top-level\footnote{Such if-else constructs is restricted to being at the top of the function's AST.} \textsc{if-then-else} unification operator where the condition can be any boolean expression in the hypothesis space. The subsequent synthesis algorithms following this approach do {\em not} compute $|\programsp|$ directly, or make it straightforward to compute it. We design a synthesis algorithm based on this approach, and a procedure to soundly compute
an upper bound on $|\programsp|$.  We choose our target language as the \sygus string-manipulating program DSL~\cite{sygus-benchmark}. Our synthesis algorithm is referred to as \gstun.

\paragraph{Vanilla \gstun: Overview.} \gstun instantiates the previously proposed approach of bottom-up explicit search with observational equivalence reduction~\cite{albarghouthi2013recursive}. 
%\gstun can be invoked on a given set of I/O examples. 
Its algorithm consists of two phases at a high level: an enumeration phase and a unification phase. In the first phase, the synthesizer enumerates
candidate programs only by repeatedly using function application.
It clusters all candidate programs which have the same I/O behavior
on the given examples and saves only one program representative of each cluster. Such enumerated programs may only be consistent with subsets 
of the given I/O examples.
In the unification phase, \gstun composes enumerated programs with an \textsc{if-then-else} unification operator. The final synthesized program $P$, therefore, can be a straight-line program (obtained by repeated function application) or a program with nested compositions of the form $\texttt{if}~P_1\text{ } \texttt{then}~P_2 \text{ } \texttt{else}~P_3$, where $P_2$ and $P_3$ can be nested programs themselves. The nesting depth is bounded internally to limit the search space. We explain the constructional details of these phases next, and explain how to compute $|\programsp|$ from the internal data-structures later. In what follows, we denote inputs and outputs of the given I/O examples as vectors $\vec{w}$ and $\vec{o}$ respectively.

\paragraph{Vanilla \gstun: Enumeration.}
\gstun enumerates all candidate programs in a bottom-up fashion by generating programs through function application. We start with the smallest syntactic programs, which are just single components (or syntactic terminals) in the target language, working up to programs with more than one component. For instance, \texttt{concat(input0, input1)} is a candidate program with three components. The total hypothesis space without conditionals is fixed based on a user-provided constraint on the maximum component size, specified as the maximum number of components the straight-line program contains. For each program
created in the enumeration phase, we compute the outputs of the program on the given
I/O examples. Note that this step does not require explicitly individually creating and running all candidate programs---it is possible to evaluate outputs during the bottom-up construction of the programs~\cite{albarghouthi2013recursive}.

We compute $2$ useful data structures internally during enumeration. The first is a {\em consistency vector} $\vec{c}$ which captures whether an enumerated program $P$ is consistent with the given I/O examples represented by vector $\vec{w}$ and vector $\vec{o}$. Specifically, the consistency vector $\vec{c}$ for program $P$ has the $i^{th}$ element set to $\checkmark$ if the $P(\vec{w}[i]) = \vec{o}[i]$, namely, the output of $P$ on
the $i^{th}$ given input example matches the corresponding given output example. Otherwise, $\vec{c}[i]$ is set to '$\times$'. This data structure speeds up the search in two ways, conceptually. First, it is calculated on top of {\em observational equivalence}~\cite{albarghouthi2013recursive}.  If many candidate programs generate the same outputs on the given input examples, then they are all observationally equivalent, and we only need to keep one such program that has the outputs for completeness~  \cite{albarghouthi2013recursive}, thus the  consistency vector is only calculated once for those programs. Secondly, even if two programs are observationally not equivalent, and both give different incorrect outputs
for an input example, they will have the same value ($\times$) in their consistency vectors. Thus, we can effectively cluster many programs that are observationally non-equivalent but have the same consistency vector to speed up the next phase.

The second useful data structure is cluster map $\phi$. It maps consistency vectors to sets of programs. Each distinct consistency vector $\vec{c}$ computed during the enumeration phase is mapped to a set of programs $\phi(\vec{c})$ that have I/O behavior captured by $\vec{c}$.

\paragraph{Vanilla \gstun: Unification.}
In this phase, \gstun synthesizes programs with nested \textsc{if-then-else} structures. The goal is to create programs that are consistent with larger subsets of I/O examples than enumerated programs, and ideally, correct on the full set of I/O examples. The cluster map $\phi$ allows us to quickly find programs which match certain subsets of all the given I/O examples, specified by a consistency vector value. When a program $P := \texttt{if}~P_1\text{ } \texttt{then}~P_2 \text{ } \texttt{else}~P_3$ is synthesized during unification, we must carefully construct a semantically correct consistency vector for $P$, using those for sub-programs $P_2$ and $P_3$. Here, note that $P_1$ needs to be a program that evaluates to a boolean value on a given input example, say $\vec{w}[i]$. If it evaluates to \texttt{true}, then the program $P_2$ must be correct on $\vec{w}[i]$ for $P$ to be correct on $\vec{w}[i]$.
Therefore, in this case, we mark $P$ as consistent with $\vec{w}[i]$ if and only if $\vec{c}[i]$ for $P_2$ has a $\checkmark$. Analogously, if $P_1$ evaluates to \texttt{false} on $\vec{w}[i]$, then $\vec{c}[i]$ for $P_3$ should be set to $\checkmark$ for $P$ to be marked consistent with $\vec{w}[i]$; otherwise $P$ is marked inconsistent with $\vec{w}[i]$.

The above-described unification procedure synthesizes all programs with
a nesting depth of up to a pre-configured maximum (default of $2$). The nesting depth controls the hypothesis space desired by the user. The cluster map $\phi$ is updated continuously with new consistency vectors discovered in the unification phase. A successful solution is a program that matches all given examples, i.e., has a consistency vector with a $\checkmark$ for all values.

\paragraph{Computing the $|\programsp|$}. The vanilla \gstun algorithm can be slightly modified and augmented with rules shown in Table~\ref{tab:stun-count-rules} to compute the $|\programsp|$ soundly. Notice that in vanilla \gstun, when employing the observational equivalence, a program with larger components size might be discarded if there is a smaller program that has the same {\em value vector}
\footnote{The value vector
$\vec{v}$ for a program $P$ is simply the outputs of $P$ on the given
input examples, i.e., $\vec{v}[i] := P(\vec{w}[i])$ for all $i$.} 
\cite{albarghouthi2013recursive}.  But we need to count all programs at different components sizes.
To do so, we store multiple counting values (and representative programs) for different component sizes along with each value vector.

During the enumeration phase, programs are synthesized bottom-up from smallest
components size to larger ones. Let $t$ be the components size of a program. We keep track
of a $\texttt{Count}(\vec{v}, t)$ for each value vector $\vec{v}$ computed for programs
with size $t$. The $\texttt{Count}(\vec{v}, t)$ for $t=1$ (smallest base components) can be directly enumerated (rule 1 in Table \ref{tab:stun-count-rules}), since these are program input arguments or constant components in our target language. 
For $t > 1$, $\texttt{Count}(\vec{v}, t)$ can return $0$ if there is no enumerated program with components size $t$ that outputs value vector $\vec{v}$.
Same value vectors $\vec{v}$ may have different counts for different components sizes $t$, thus we enumerate on tuple $(\vec{v},t)$ rather than just $\vec{v}$. 
When \gstun uses function application to generate a new program $P'$ from programs $P_i$ (rule 2.1 in Table~\ref{tab:stun-count-rules}), the count for the value vector of the resulting $P'$ is updated by adding the product of all the counts of its arguments $P_i$ at their respective components sizes (rule 2.2 in Table~\ref{tab:stun-count-rules}).
This completes all the ways programs are compositionally created in the enumeration phase from component size 1 to the maximum component size.

After the enumeration on value vectors is finished, we have also clustered observationally non-equivalent programs based on a consistency vector $\vec{c}$ in $\phi$. 
Define $\psi[\vec{c}]$ as the set of all the value vectors that corresponds to a consistency vector $\vec{c}$, we sum up the $\texttt{Count}(\vec{v}, t)$ for every $\vec{v}$ in $\psi[\vec{c}]$ (rule 3 in Table~\ref{tab:stun-count-rules}). This way, we compute counts for consistency vectors.

During unification, programs of the form $P := \texttt{if}~P_1\text{ } \texttt{then}~P_2 \text{ }\\ \texttt{else}~P_3$ are composed. 
Here, counts of the consistency vectors of $P_2$ and $P_3$ have been computed
after enumeration if $P_2$ and $P_3$ are programs with nesting depth zero.
In this case, the $\texttt{Count}(\vec{c}, 1)$ of the program $P$ is the the product of 
$P_1$, all possible $P_2$ (0 condition) with $P_1$ as the condition, and all possible $P_3$ (0 condition) with $P_1$ as the condition, summed over all possible $P_1$. Thus, we have computed \texttt{Count} for the consistency vectors
of programs with nesting depth of $1$. Using this, we can recursively  compute counts for consistency vectors of programs with nesting depth $2$ or more (rule 5 in Table~\ref{tab:stun-count-rules}). To optimize, we memorize counts
of individual consistency vectors as well as for sets
of consistency vectors. For example, the set comprising two consistency vectors $\langle \checkmark, \times, \times \rangle$ and $\langle \checkmark, \times, \checkmark \rangle$ is succinctly represented as $\langle \checkmark, \times, \top \rangle$, and their sum of counts is memorized (rule 4 in Table~\ref{tab:stun-count-rules}). We prove that the rules in Table~\ref{tab:stun-count-rules} provide a sound upper bound of the size of the hypothesis space $|\programsp|$ in supplementary material~\cite{aaasup}.

\newcommand{\countv}{\text{Count}}
\newcommand{\cvset}{\mathcal{C}}
\newcommand{\countcset}{\text{Count}}
\newcommand{\countc}{\text{Count}}
\newcommand{\setcondb}{\mathcal{B}}
\newcommand{\condsplitthen}{\Gamma_{\text{then}}}
\newcommand{\condsplitelse}{\Gamma_{\text{else}}}

% --------------------
\newcommand{\myrowpadding}{\vspace{4pt}}
% Define typographic struts, as suggested by Claudio Beccari
%   in an article in TeX and TUG News, Vol. 2, 1993.
\newcommand\Tstrut{\rule{0pt}{2.6ex}}         % = `top' strut
\newcommand\Bstrut{\rule[-0.9ex]{0pt}{0pt}}   % = `bottom' strut
% --------------------

\begin{table}
    \caption{Count rules for \gstun. 
    We define the set of all boolean value vectors as $\setcondb$, and we use $\cvset$ to represent a succinct representation of consistency vectors, which can be a singleton or a set of consistency vectors.
    The \texttt{Count} value is calculated differently for value vector $\vec{v}$, consistency vector $\vec{c}$, or their succinct representation $\cvset$.
    }
    \footnotesize
    \label{tab:stun-count-rules}
    \centering
    \begin{tabularx}{\columnwidth}{ l } 
    \hline
    \textbf{Count Rules for Enumeration}\Tstrut\\
    \tabincell{l}{
    Starting point of enumeration:\\
    \begin{math}
    \countv(\vec{v}, 1) = \text{number of single components that output }\vec{v}
    \end{math}\\
    \begin{math}
    \countv(\vec{v}, t) = 0, \text{for } t > 1.   \quad \quad (1)
    \end{math}\\ \\
    When program $P'$ with value vector $\vec{v}'$ is enumerated at component size $t'$:\\
    Let, $P'$ = $f(P_1, P_2, ...,P_n)$ for some function component $f$ and programs $P_i$.\\
    If $t_i$ is the component size for $P_i$, \\
    and $\vec{v}_i$ is the  value vector of ($P_i$, $t_i$), $i\in\{1,...,n\}$,\\ 
    then
    $ 1 + \sum t_i = t'$ holds, and \quad \quad (2.1) \\
    \begin{math}
    \begin{aligned}
    \countv(\vec{v}', t')\leftarrow \countv(\vec{v}', t') + \prod_{i=1}^n \countv(\vec{v}_i,t_i) \quad \quad (2.2) \\ 
    % \vec{v}_f = \texttt{eval}(f, (\vec{v}_1, \vec{v}_2, ..., \vec{v}_n)), \quad \quad (2)
    \end{aligned}
    \end{math}
    }
    \myrowpadding
    \\
    \hline
    \textbf{Count Rules for Clustering}\Tstrut\\
    \tabincell{l}{
    \begin{math}
         \text{Count} (c) = \sum_{\vec{v}\in \psi[c]}\sum_t(\countv(\vec{v},t))  \quad \quad (3)
    \end{math}
    }
    \myrowpadding
    \\
    \hline
    \textbf{Count Rules for Unification}\Tstrut\\
    \tabincell{l}{
    \begin{math}
    \cvset_{\text{goal}} = \langle\checkmark, \checkmark, ..., \checkmark\rangle,
    \end{math}\\
    \begin{math}
    \text{ the count of hypothesis space up to }k\text{ conditions is } \sum_{i=0}^k \countcset(\cvset_{goal}, i)
    \end{math}\\
    \begin{math}
    \countcset(\cvset, 0) = \sum_{c\in \cvset} \countc(c) 
    \quad \quad (4)
    \end{math}\\
    \begin{math}
    \begin{aligned}
    \countcset(\cvset, i) =& \sum_{\vec{b}\in \setcondb}
    \sum_{j=0}^{i-1} (\countcset(\condsplitthen(\cvset, \vec{b}), j) \times\\
    &\countcset(\condsplitelse(\cvset, \vec{b}), i-j-1)\times ~\sum_t\countv(\vec{b}, t)) \quad \quad (5)\\
    \end{aligned}
    \end{math}\\
    \begin{math}
    \begin{aligned}
    \quad \text{where } &  \condsplitthen(\cvset, b) \text{ and } \condsplitelse(\cvset, b)
    \text{ are also succinct }\\
    &\text{ representation of consistency vectors, and} \\
    &\condsplitthen(\cvset, b)_i = 
    \begin{cases}
        \text{\checkmark} \text{  if } b_i = T\\
        \top \text{  otherwise}
    \end{cases},
    \condsplitelse(\cvset, b)_i = 
    \begin{cases}
        \text{\checkmark} \text{  if } b_i = F\\
        \top \text{  otherwise}
    \end{cases}
    \end{aligned}
    \end{math}\\
    }
    \myrowpadding
    \\
    \hline
    \end{tabularx}
\end{table}

\paragraph{\gstun Augmented with \tool.}
We call \gstun augmented with \tool as described in Section~\ref{sec:design-synth} as \toolstun. The maximum number of nested conditions during the unification phase of \gstun is $2$, and without much loss of expressiveness, we only allow the \texttt{else} branch to have nesting. With this setting, the hypothesis space of \gstun is a union of: 
\begin{itemize}
    \item \textbf{$H_0$} (The set of straight-line programs),
    \item \textbf{$H_1$} (The programs of form $\texttt{if}~P_1\text{ } \texttt{then}~P_2 \text{ } \texttt{else}~P_3$),
    \item \textbf{$H_2$} (The programs of form $\texttt{if}~P_1\text{ } \texttt{then}~P_2 \text{ } \texttt{else}~\texttt{if}~P_3\text{ } $
    \\$\texttt{then}~P_4 \text{ } \texttt{else}~P_5$).
\end{itemize}

From $H_0$ to $H_2$, the hypothesis space is increasingly expressive. \toolstun
invokes the \tool loop with $H_0$ first, and if it returns None then with $H_1$, and so on in that order. This has the nice property that it will return $f$ consistent with existing examples from $H_i$ where $i$ is the smallest possible. For correctness, each invocation with a new hypothesis $H_i$ uses a failure probability of $\frac{\delta}{3}$, so the total failure probability is bounded by $\delta$ (union bound).

\section{Evaluation}

We have shown that when an existing PBE synthesizer using \tool returns a synthesized program, the program generalizes, i.e., it is close to the target with high probability. Our evaluation focuses on two empirical utility goals in string-manipulation tasks:

\begin{enumerate}
\item \textbf{Accuracy:} Do our theoretical generalization guarantees improve the end accuracy of existing PBE synthesizers?

\item \textbf{Sample Size:} How many examples does \tool require to achieve provable generalization?

\end{enumerate}

Recall that \tool primarily extends existing synthesizers to control how many examples the synthesizer sees before stopping. 
We evaluate (a) \toolvsa, which builds on the PROSE framework, and (b) \toolstun, which is implemented on the \gstun synthesizer we designed. The vanilla \gstun synthesizer is around 4000 lines of C++ code. These vanilla versions of PBE synthesizers (without the \tool augmentation) serve as our baselines to measure improvements due to \tool.

We point out that PBE synthesizers for string programs often compete on computational overheads reported for producing {\em any} program that fits a given set of examples. Our accuracy criterion and our objective are completely
different---we want to check when a synthesizer produces a program close to a fixed
target program (the number of examples is not fixed). This is why we do not compare
to other baseline solvers which may be computationally faster~\cite{reynolds2019cvc,reynolds2015counterexample}, but are not designed to generalize to a target program.

\paragraph{Benchmarks.}
For \toolvsa, we considered $16$ common string-related programming tasks as target functions to synthesize. These are of similar style and complexity as those reported in FlashMeta paper~\cite{polozov2015flashmeta} such as changing the date format, extracting numbers or abbreviating words. Henceforth, we refer to these programs as \vsab, details of which are in the supplementary material~\cite{aaasup}. For \toolstun, we take the euphony benchmark from the PBE-Strings track of the \sygus 2019 benchmark~\cite{sygus-benchmark}\footnote{downloaded from \url{https://github.com/SyGuS-Org/benchmarks/tree/master/comp/2019/PBE_SLIA_Track/euphony}} which
contains $100$ PBE tasks with $2-16$ examples. 

The target programs are not available for those $100$ tasks, so we manually wrote them from the given examples from the benchmark. Out of the $100$ tasks, $10$ are for tasks that output boolean values which are not in the scope of our considered DSL. Further, we experimentally observed that our \gstun implementation scales up to component size $9$ (size of the longest straight-line programs before unification) within a reasonable computation of a day for all benchmarks to finish on our experimental setup (larger component size increases the program search space). So for the remaining $90$, we filtered out the ones that could not be manually constructed under component size $9$. This finally results in $59$ \sygus benchmarks which we call \stunb.

%\paragraph{Experimental setup.}
The generalization error tolerance for all experiments is set to $5\%$
($\epsilon=0.05$) and confidence parameter to $98\%$ ($\delta=0.02$) by
default. When comparing sample size for different $\epsilon$, we also run the
benchmark on ($\epsilon=0.02, \delta=0.02$) and ($\epsilon=0.1, \delta=0.02$). 
The default step size $k$ for the sampling phase is set to $1$ for \toolvsa and $20$ for \toolstun.
With these parameters, it is guaranteed with 
probability at least $98\%$ that when \tool stops, its synthesized program is going to have a generalization error of at most $5\%$.
We ran all experiments on Amazon EC2 Ubuntu 16.04 instance with 512GB RAM, 64-core 3.1GHz Intel Xeon processors where each benchmark runs $1$ core. All our experiments finished within 24 hours. For \toolvsa, 79\% runs finished in 1 minute and 100\% within 1 hour. For \toolstun, 75\% runs finished in 10 minutes and 97\% within 2 hours.

%\paragraph{Input Sampling.}
\tool works with any input distribution for creating I/O examples. We choose a distribution that is easy to generate and not specialized to each target program. Specifically, we simulate a black-box fuzzer for string inputs. All our \toolvsa evaluation reports an average over $32$ trials of each target program, and for each trial, we sample a string input as follows:
\begin{itemize}
    \item the string length is chosen uniformly at random from $8-16$;
    \item each character in the string is either chosen uniformly at random from the character set $\textsf{C}=\texttt{"A-Za-z0-9,.-;|"}$, or chosen as white-space with probability $15\times$ larger than the probability of any character in \textsf{C}.
\end{itemize}
We run each such input created on the target program to create the output. The input-output pairs are given to the synthesizer.

We evaluate \toolstun over $3$ trials for each target, as these programs are computationally heavier to synthesize. For each trial, we simulate a basic mutation-based fuzzer as the sampling distribution. Specifically, we take the input strings provided in the \sygus benchmark as seeds, and mutate them randomly as follows:
\begin{itemize}
    \item add a string of randomly chosen length up to size $10$ at a randomly chosen location of the seed string, with each character being a printable value\footnote{In Python, we use \texttt{string.printable} and remove the white space characters}; or
    \item remove a randomly chosen character from the seed string.
\end{itemize}

For programs with more than one input argument, we additionally add rules specifying whether one of the inputs is a substring of another input argument. We then randomly choose a substring from that input argument when synthesizing examples. For integer inputs, we randomly chose either an integer bounded by the length of one of the string input arguments or a randomly chosen integer from $(0,1000)$.
This ensures that the target program can be run on the mutated inputs without resulting in type errors or failure.

\subsection{Accuracy Improvement}
 
To evaluate whether our theoretical generalization guarantees translate into improved correctness, we check that the synthesized program is {\em correct}, i.e., syntactically or semantically equivalent to the target program. Our syntactic equivalence is confirmed automatically, and for semantic equivalence, we resort to manual inspection.
While \tool might produce programs that are close to the target as per its $\epsilon$-close guarantee, it is difficult to estimate closeness objectively. Therefore, we take a conservative approach and only report whether the synthesized program is correct. Programs that are "almost" or "close to" correct are reported as incorrect.

\toolvsa synthesizes $14/16$ programs semantically equivalent to the target program in all $32$ runs for $\epsilon=0.05,\delta=0.02$ which shows that \toolvsa is useful to synthesize common string manipulation programs. For one of the remaining benchmarks, the synthesized program is correct on $26/32$ runs. For $1$ benchmark, the synthesized program is correct on $7/32$ runs.
In total, \toolvsa produces correct programs in $481/512$ $(93.95\%)$ runs. 

We observe that the vanilla \gprose synthesizer (without \tool) consistently overfits when sample sizes are chosen arbitrarily smaller than mandated by \toolvsa. For instance, when we use exactly $4$ randomly chosen examples in each trial, the vanilla \gprose synthesizer produces incorrect programs on most of the $32$ runs for all target programs. Most of the synthesized programs overfit the examples: it is only correct on $176/512$ ($34.38\%$) runs in total. This confirms the importance of \tool's main objective: taking enough examples until the synthesized program is guaranteed to generalize with high probability. 

\toolstun synthesizes $53/59$ correct programs from the \stunb for $\epsilon=0.05,\delta=0.02$ for all $3$ runs. In total, this leads to $159/177$ ($89.83\%$) correct runs.
As a point of comparison, the vanilla version \gstun synthesizer, evaluated on the examples provided in the \sygus benchmark, produces correct programs for $36$/$59$ of the target benchmark.
\toolstun shows a $29\%$ improvement over the vanilla \gstun on the \sygus benchmark, synthesizing correctly an additional $17$ programs in all trials.
Further, this suggests that a significant number of \stunb programs in the benchmark do {\em not} have enough examples in the benchmark to provably generalize. Synthesizers, therefore, may need additional hints or assumptions to solve them correctly.

To analyze vanilla \gstun under the same input distribution as \toolstun, we further evaluate it on a fixed number of randomly chosen examples in $3$ trials. We use sample size of $4$ per trial following prior work~\cite{singh2015predicting,gulwani2011automating}. We find that vanilla \gstun synthesizes only $33/59$ correct programs in all $3$ runs (in total correct on $121/177$ runs), confirming that it often overfits.

\subsection{Sample Size Sufficient for Generalization}
\label{sec:eval-samplesize}

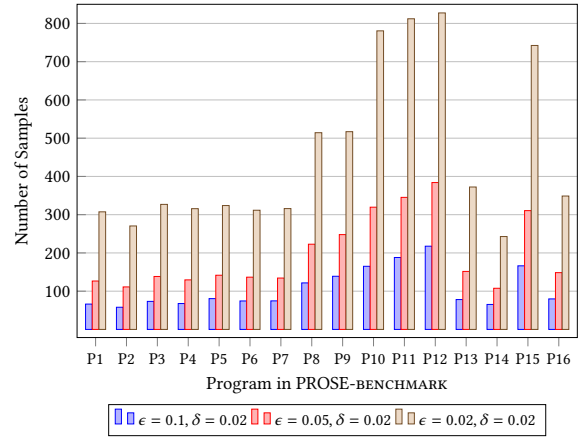
\begin{figure}
\centering
\resizebox{0.9\linewidth}{!}{%
\begin{tikzpicture}
\begin{axis}[
        yminorticks=false,
        xtick pos=left,
        ytick pos=left,
        % adjust the `width' a bit by keeping the default `height'
        width=1.2*\axisdefaultwidth,
        height=\axisdefaultheight,
        ylabel={Number of Samples},
        xlabel={Program in \vsab},
        % set appropriate `ymax' value so the `nodes near coords' fit to the plot
        % (adjusting the `ymin' value is just to make it look a little bit better)
        % there should be no gap between the bars in one group
        ybar=0pt,
        ytick={100,200,300,400,500,600,700,800},
        ymax=850,
        % use data from the table for the xticklabels
        xtick=data,
        xticklabels from table={\loadedtable}{id},
        x tick label style={font=\small},
        legend style={font=\small,at={(0.5,-0.2)}, anchor=north, legend columns=3},
        ymajorgrids,
        % to start the bars from the bottom edge of the plot
        % (otherwise they would start from 10^0
        %  borrowed from <http://tex.stackexchange.com/a/86688/95441)
        log origin=infty,
        % adjust the size of the bars so they don't overlap
        % (you can play with the numerator to change the gap between the groups)
        bar width=0.65/\NoOfCols,
        % enlarge the x limits so all of the bars are shown
        % (play with the value to adjust the gap on the sides of the plot)
        enlarge x limits={abs=0.6},
        % and position the legend outside of the plot to not overlap with the data
        % legend pos=outer north east,
        % add `nodes near coords'
        %nodes near coords={
        %    % because internally PGFPlots works with floating point numbers, we
        %    % change them to fixed point numbers
        %    \pgfkeys{
        %        /pgf/fpu=true,
        %        /pgf/fpu/output format=fixed,
        %    }%
        %    % check if numbers are greater than 1000 and if so, divide them by
        %    % 1000 to convert them from ms to s scale
        %    \pgfmathparse{
        %        ifthenelse(
        %            \pgfplotspointmeta < 1000,
        %            \pgfplotspointmeta,
        %            \pgfplotspointmeta/1000
        %        )
        %    }%
        %    % to now decide which of the two cases we have, we compare the
        %    % point meta value, but because `\ifnum' compares integers, we first
        %    % have to convert the fixed number to an integer
        %        \pgfmathtruncatemacro{\Y}{\pgfplotspointmeta}%
        %    \ifnum\Y<1000
        %        \pgfmathprintnumber{\pgfmathresult}
        %    \else
        %        \pgfmathprintnumber{\pgfmathresult}
        %    \fi
        %},
        % set the style of the `nodes near coords'
        %nodes near coords style={
        %    font=\tiny,
        %    rotate=90,
        %    anchor=west,
        %},
        %% as basis for the `nodes near coords' use the raw y values
        %point meta=rawy,
    ]
        % add the data rows
        \foreach \i in {1,...,\NoOfCols} {
            \addplot table [
                x expr=\coordindex,
                y index=\i,
                col sep=comma,
            ] {\loadedtable};
                % to automatically add the legend entries first extract the
                % column name and store it in `\colname'
                % (this is an undocumented command so far. I borrowed it from
                %  <http://tex.stackexchange.com/q/171021/95441>)
                    \pgfplotstablegetcolumnnamebyindex{\i}\of{\loadedtable}\to{\colname}
                % now you can add the legend entry
                % (because we are in a loop we have to use the "expanded" version)
                %\addlegendentryexpanded{\colname};
        }
      \addlegendentry{$\epsilon=0.1,\delta=0.02$}
      \addlegendentry{$\epsilon=0.05,\delta=0.02$}
      \addlegendentry{$\epsilon=0.02,\delta=0.02$}
    \end{axis}
\end{tikzpicture}
}
\caption{For most programs in the \vsab, \toolvsa synthesizes programs with provable generalization
  under $400$ examples for $\epsilon=0.02,\delta=0.02$. 
  For $\epsilon=0.1,\delta=0.02$, the \#samples drop to $58-218$ (average $107.22$).}
\label{fig:vsa-sample-count}
\end{figure}

\begin{figure}
\centering
\resizebox{0.9\linewidth}{!}{%
\pgfplotstableread[col sep=comma]{data/sgstun_avg_sample_size.csv}\csvtable
\begin{tikzpicture}
\begin{axis}[
        yminorticks=false,
        xtick pos=left,
        ytick pos=left,
        % adjust the `width' a bit by keeping the default `height'
        width=1.2*\axisdefaultwidth,
        height=\axisdefaultheight,
        ylabel={Number of Samples},
        xlabel={Program in \stunb},
        % set appropriate `ymax' value so the `nodes near coords' fit to the plot
        % (adjusting the `ymin' value is just to make it look a little bit better)
        % there should be no gap between the bars in one group
        ybar=0pt,
        ytick={100,200,300,400,500,600,700,800,900,1000,1100,1200,1300,1400,1500},
        ymin=0,
        ymax=1500,
        % use data from the table for the xticklabels
        xtick={1, 5,...,60},
        %xticklabels from table={\csvtable}{pid},
        x tick label style={font=\small}, %anchor=east, rotate=90},
        legend style={font=\small,at={(0.5,-0.2)}, anchor=north, legend columns=3},
        ymajorgrids,
        % to start the bars from the bottom edge of the plot
        % (otherwise they would start from 10^0
        %  borrowed from <http://tex.stackexchange.com/a/86688/95441)
        log origin=infty,
        % adjust the size of the bars so they don't overlap
        % (you can play with the numerator to change the gap between the groups)
        bar width=2.2pt,
        % enlarge the x limits so all of the bars are shown
        % (play with the value to adjust the gap on the sides of the plot)
        enlarge x limits={abs=0.6},
        % and position the legend outside of the plot to not overlap with the data
        % legend pos=outer north east,
        % add `nodes near coords'
        %nodes near coords={
        %    % because internally PGFPlots works with floating point numbers, we
        %    % change them to fixed point numbers
        %    \pgfkeys{
        %        /pgf/fpu=true,
        %        /pgf/fpu/output format=fixed,
        %    }%
        %    % check if numbers are greater than 1000 and if so, divide them by
        %    % 1000 to convert them from ms to s scale
        %    \pgfmathparse{
        %        ifthenelse(
        %            \pgfplotspointmeta < 1000,
        %            \pgfplotspointmeta,
        %            \pgfplotspointmeta/1000
        %        )
        %    }%
        %    % to now decide which of the two cases we have, we compare the
        %    % point meta value, but because `\ifnum' compares integers, we first
        %    % have to convert the fixed number to an integer
        %        \pgfmathtruncatemacro{\Y}{\pgfplotspointmeta}%
        %    \ifnum\Y<1000
        %        \pgfmathprintnumber{\pgfmathresult}
        %    \else
        %        \pgfmathprintnumber{\pgfmathresult}
        %    \fi
        %},
        % set the style of the `nodes near coords'
        %nodes near coords style={
        %    font=\tiny,
        %    rotate=90,
        %    anchor=west,
        %},
        %% as basis for the `nodes near coords' use the raw y values
        %point meta=rawy,
    ]
        % add the data rows
    \addplot table [x expr=\coordindex,y={sample_avg}] {\csvtable};
    \addlegendentry{$\epsilon=0.05,\delta=0.02$}
    \end{axis}
\end{tikzpicture}
}
\caption{For most programs in the \stunb, \toolstun synthesizes programs with
provable generalization under $500$ examples for $\epsilon=0.05,\delta=0.02$.
Only $11$ of these programs require $500-1400$ examples.}
\label{fig:stun-sample-count}
\end{figure}
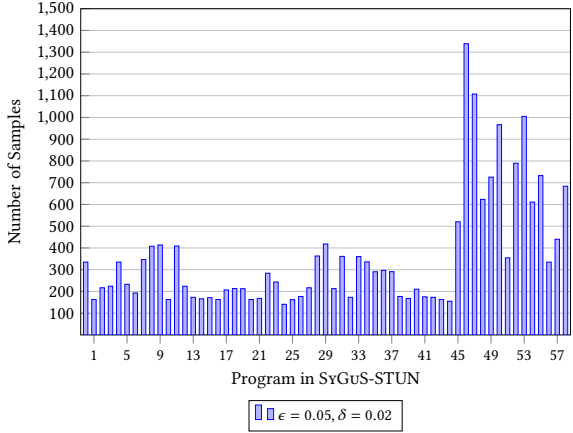

Our work provides empirical evidence that a modest number of examples suffice for provable generalization for the evaluated tasks.
Figure~\ref{fig:vsa-sample-count} shows that we need between $100$-$400$ examples (about $197$ on average) to
achieve $(\epsilon=0.05, \delta=0.02)$ generalization for the \toolvsa on the $16$
target programs evaluated. For a smaller error tolerance $(\epsilon=0.02, \delta=0.02)$, the observed
range of sample size becomes $200$-$900$. Note that sample size is sensitive to
$\epsilon$ as is theoretically expected---distinguishing between two functions
that behave almost identically using random sampling will require many samples.

Figure~\ref{fig:stun-sample-count} shows that we need between $140$-$1400$ examples (about $357$ on average) to
achieve $(\epsilon=0.05, \delta=0.02)$ generalization for the $59$ target programs evaluated with \toolstun. Most of the programs in the \sygus benchmark require under $500$ samples, with only $11$ programs requiring $500-1400$ I/O examples.

\subsection{Reduction in Program Search Space}
\label{sec:eval-hypspace}

Note that Theorem~\ref{thm:static-union-bound} does {\em not} predict how fast the procedure will converge to a generalized program, i.e., how fast the space of consistent programs will shrink after each example. This varies empirically given the task and examples seen. But, \tool internally estimates (conservatively) how many programs remain consistent after seeing each example. From this, two empirical findings which explain our other observations emerge. First, the program space shrinks {\em drastically} with the first few examples for nearly all benchmarks. Second, the number of examples required to generalize depends significantly on the expressiveness of the chosen hypothesis space.

\begin{figure}
  \centering
  \resizebox{0.75\linewidth}{!}{%
  \begin{tikzpicture}
  \begin{semilogyaxis}
    [
        xmin=-0.5,
        xmax=28,
        ymin=0.5,
        ymax=1e51,
        yminorticks=false,
        xtick pos=left,
        ytick pos=left,
        xlabel={Number of Samples},
        ylabel={Program Space Size},
        legend style={font=\small,at={(0.5,-0.2)}, anchor=north, legend columns=6},
    ]
    %\addplot table [x=id, y=P1, col sep=comma] {data/vsa_h_drop.csv};
    \addplot table [x=id, y=P2, col sep=comma] {data/vsa_h_drop.csv};
    %\addplot table [x=id, y=P3, col sep=comma] {data/vsa_h_drop.csv};
    %\addplot table [x=id, y=P4, col sep=comma] {data/vsa_h_drop.csv};
    \addplot[mark=x,color=black] table [x=id, y=P5, col sep=comma] {data/vsa_h_drop.csv};
    %\addplot table [x=id, y=P6, col sep=comma] {data/vsa_h_drop.csv};
    %\addplot[mark=triangle*, color=red]  table [x=id, y=P7, col sep=comma] {data/vsa_h_drop.csv};
    %\addplot table [x=id, y=P8, col sep=comma] {data/vsa_h_drop.csv};
    %\addplot table [x=id, y=P9, col sep=comma] {data/vsa_h_drop.csv};
    %\addplot table [x=id, y=P10, col sep=comma] {data/vsa_h_drop.csv};
    \addplot[mark=x,color=blue] table [x=id, y=P11, col sep=comma] {data/vsa_h_drop.csv};
    \addplot[mark=otimes, color=black] table [x=id, y=P12, col sep=comma] {data/vsa_h_drop.csv};
    \addplot[mark=diamond, color=brown] table [x=id, y=P13, col sep=comma] {data/vsa_h_drop.csv};
    %\addplot[mark=triangle*, color=black] table [x=id, y=P14, col sep=comma] {data/vsa_h_drop.csv};
    \addplot[mark=otimes, color=red] table [x=id, y=P15, col sep=comma] {data/vsa_h_drop.csv};
    %\addplot[mark=*, color=blue] table [x=id, y=P16, col sep=comma] {data/vsa_h_drop.csv};

    %\legend{P1,P10,P11,P12,P13,P14,P15,P16, P2,  P3,  P4,  P5,  P6,  P7,  P8,  P9}
    %\legend{P1, P2, P3, P4, P5, P6, P7, P8, P9, P10, P11, P12, P13, P14, P15, P16}
    %\legend{P3,  P4,  P5,  P7,  P9, P12}
    %\legend{P11, P12, P13, P15, P2, P5}
    \legend{P2, P5, P11, P12, P13, P15}
    \end{semilogyaxis}
\end{tikzpicture}
}
\caption{In \vsab, the program space shrinks $3.7\cdot10^2-10^{44}\times$ on average with the first $5$ examples, explaining why \toolvsa can provably generalize in modest number of samples for most programs in the benchmark.
}
\label{fig:vsa-hypothesis-space}
\end{figure}
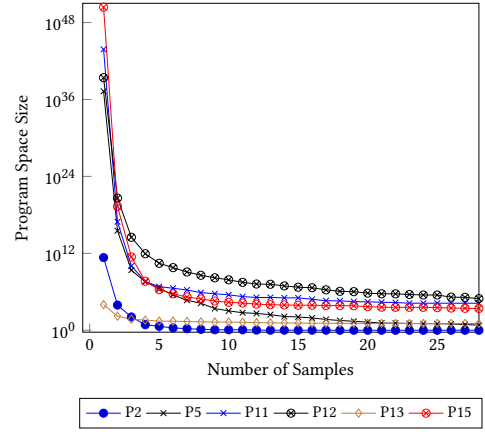

% start with main finding
Figure~\ref{fig:vsa-hypothesis-space} shows the size of the consistent program space, for $6$ representative programs in the \vsab computed by \toolvsa.  The Y-axis is a logarithmic scale. The full evaluation of the rest of the programs is in the supplementary material~\cite{aaasup}.
We choose these programs as they represent the largest, smallest, and average cases of the program space size after $25$ examples averaged over $32$ runs, as well as the largest and smallest program space size decrease after $5$ examples. In particular, P$15$ has the largest decrease on average in the program space size after $5$ examples from $10^{50}$ to $10^{6}$ %($\approx2\cdot10^{41}\times$) 
while P$13$ has the smallest decrease after $5$ examples from $9.7\cdot10^{3}$ to $27$. %($\approx7.3\cdot10^{2}\times$).
This observation explains why a small number of examples turn out to
be sufficient for generalization in this benchmark. It also shows
that reducing the consistent hypothesis space further, after the initial quick reduction,
becomes increasingly difficult with unbiased sampling.

The \sygus benchmark has target programs of 
different complexity (different number of conditionals). 
For \toolstun, the number of samples required to generalize depends
on target program's complexity---more complex programs require considering a larger original hypothesis space to be represented. However, the algorithm does not know the original hypothesis space. To use fewer samples, the algorithm chooses the smallest hypothesis space that still contains programs consistent with samples, because outputting more complex programs (more conditionals) requires choosing a larger hypothesis space for which more samples are needed. 

We use the target programs in $H_0$ as an example to show this phenomenon in Figure \ref{fig:stun-hypothesis-space}. The figure shows for each target program the sample size sufficient for $(\epsilon, \delta/3)$-generalization calculated by the 3 parallel \tool instances on $H_0$, $H_1$, and $H_2$.
We show that if \toolstun chooses a program in $H_0$, the number of examples is $141-419$. If the target program is in $H_0$ but the synthesizer chooses a program in $H_1$, the number of samples is larger by $1166.22$ on average. For example, for program $20$ in our $3$ runs, $213$ samples are sufficient to pick a program in $H_0$, but to return a program from $H_1$ or $H_2$ \toolstun requires around $1300$ and $2400$ examples, respectively.
This quantitatively shows that the sample size can vary by a large margin when considering more complex programs. 
Moreover, this result explains why choosing a simpler program first can require a smaller number of examples.

\begin{figure}
  \centering
  \resizebox{0.75\linewidth}{!}{%
  \begin{tikzpicture}
  \begin{axis}
    [
        xmin=0,xmax=46,
        ymin=0,
        ymax=3200,
        ymajorgrids,
        xtick={1, 5, ...,45},
        %at={(-0.3\textwidth,0)},
        %xticklabels from table={\csvtable}{pid},
        x tick label style={font=\small}, %anchor=east, rotate=90},
        yminorticks=false,
        xtick pos=left,
        ytick pos=left,
        xlabel={Program in \stunb},
        ylabel={Number of Samples},
        legend style={font=\small,at={(0.5,-0.2)}, anchor=north, legend columns=6},
    ]
    \addplot[only marks,mark=diamond*,mark options={fill=blue}] table [x expr=1+\coordindex, y={sample_H0}, col sep=comma] {data/stun_avg_sample_h0h1h2_first.csv};
    \addplot[ only marks,mark=*,mark options={fill=red}] table [x expr=1+\coordindex, y={sample_H1}, col sep=comma] {data/stun_avg_sample_h0h1h2_first.csv};
    \addplot[only marks,mark=triangle*,mark options={fill=brown},ycomb] table [x expr=1+\coordindex, y={sample_H2}, col sep=comma] {data/stun_avg_sample_h0h1h2_first.csv};
    \legend{$p \in H_0$, $p \in H_1$, $p \in H_2$}
    \end{axis}
    \end{tikzpicture}
}
\caption{For target programs with no condition ($t\in H_0$), choosing a program $p \in H_0$, versus a program with one condition ($H_1$) or two conditions ($H_2$) leads to provable generalization in less number of samples. }
\label{fig:stun-hypothesis-space}
\end{figure}
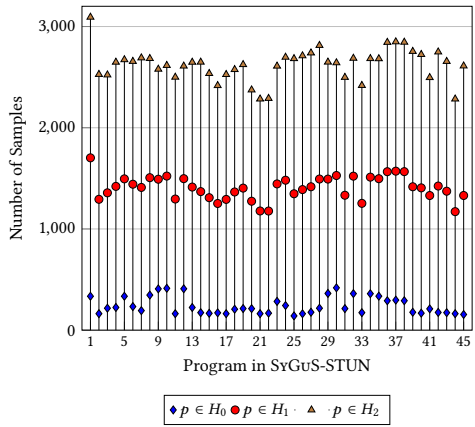

\section{Related Work}

The overfitting problem in learning programs from examples is known. Many different approaches have been proposed to tackle it (see Section~\ref{sec:intro}). One line of work proposes conditioning the search with program traces rather than just I/O examples~\cite{shin2018improving,ellis2017learning,chen2018execution}. Another line of work improves the input specification~\cite{laich2020guiding,drachsler2017synthesis} using domain-specific knowledge about the hypothesis space. Singh et al. propose to rank the synthesized functions based on distributional priors with machine learning~\cite{singh2015predicting}.  Similarly, several inductive synthesis techniques use deep learning to improve their search~\cite{balog2016deepcoder,polosukhin2018neural,zohar2018automatic}. Broadly speaking, these approaches are complementary to ours as they either require domain-specific priors or they learn program patterns from a dataset of similar programs~\cite{raychev2015Bigcode}. Moreover, none of them claims any generalization guarantees.

Several works suggest that larger test harnesses lead to promising improvements in synthesized programs~\cite{perelman2014TDS,laich2020guiding}, which is also observed in related domains of program repair~\cite{goues2019apr,xiang2020feedback} and invariant learning~\cite{blazytko2020Aurora}. Our work, motivated by these, provides the formal
bridge between test harness size and provable generalization. We also show that it translates into significantly improved accuracy on string-manipulating tasks, directly due to reduced overfitting.

Establishing generalization guarantees for synthesizers has been studied for over three decades. The PAC learnability framework has been introduced by Valiant~\cite{valiant1984PAC} for analyzing generalization from a computational perspective. Under this theory, the sample complexity required for generalization has been established for learning in propositional logic domain~\cite{valiant1984PAC,HAUSSLER199278,rivest1987DL}. The results have been extended to learning logic programs i.e., predicate logic domain~\cite{dzeroski1992pac,coheneff1994pac,cohenneg1994pac}. Some of the above bounds are limited to certain types of hypothesis spaces. Instead, Blumer \etal have given two sample complexity bounds 1) union bound~\cite{blumer1987occam} and 2) Vapnik-Chervonenkis Dimension bound~\cite{blumer1989learnability} that are more general. The union bound is used when the hypothesis space is finite and when the VC dimension of a model is difficult to estimate~\cite{lau2003programming,blanc2019topdown}. Whilst the above works have established the bounds in theory, their applicability in real-world synthesizers have been very limited. A recent work uses the VC dimension argument for synthesizing linear arithmetic functions for holes given in a sketch~\cite{drews2019efficient}. For many programming domains, such as for string-manipulating programs, it is difficult to compute VC dimensions.

Besides the PAC framework, another line of work towards generalization is through active learning. Some interactive synthesis systems have question selection mechanisms to find distinguishing input \cite{mayer2015user,wang2017interactive}. \citeauthor{ji2020question} further approximate optimal questions to resolve ambiguity in less number of samples \cite{ji2020question}. However, these approaches do not give generalization guarantees without assuming the existence of the target programs in the hypothesis space or a prior distribution over target programs, so they are orthogonal to our approach which works under minimal assumptions.

Outside of program synthesis, generalization has been extensively studied in machine learning. Our work bridges the two lines of inquiry that have evolved in parallel. Apart from PAC-style definitions based on sample complexity, generalization can be achieved using algorithmic stability~\cite{bousquet2002stability}. Bounds have been established for both convex optimizations and non-convex optimization algorithms, i.e., low sensitivity to small changes in inputs~\cite{hardt2016ICML,ben2017PAC,rivasplata2018PAC,mou2018SGLD,pensia2018Generalization}. These works leverage the properties of algorithms like stochastic gradient descent (SGD) and stochastic gradient Langevin dynamic (SGLD) in order to estimate the generalization bounds for a given number of samples. Adapting the framework of algorithmic stability to PBE-based synthesis is promising future work, but it is challenging. A direct adaptation, for example, would restrict learnt programs to be stable, for which small changes in outputs for small changes in inputs.
Generalization has been explored from other perspectives such as by bounding network capacity~\cite{neyshabur2017Gen} and overparametrization~\cite{kawaguchi2019Gen,zhu2019Gen,zou2019Gen} in machine learning literature,
which are also alternative starting points for studying generalization in program synthesis.

\section{Conclusion}

In this work, we exploit the theoretical connection between generalization and the numbers of examples used in programming-by-example synthesis. We provide the first principled approach that guarantees generalization with a modest number of examples in this regime. Key to this result is our mechanism for computing sample complexity on the fly. We show experimentally that this significantly reduces overfitting and improves accuracy for synthesizing string-manipulation programs, compared to approaches that use arbitrarily fewer examples.

\begin{acks}
  We thank Shiqi Shen, Shruti Hiray, and the anonymous reviewers for helpful feedback on this work. This work was supported by Crystal Centre at National University of Singapore,  a Singapore Ministry of Education Academic Research Fund Tier 1 (WBS number R-252-000-B50-114), and a research grant with WBS number R-252-000-B14-281. All opinions in this work are solely those of the authors.
\end{acks}

%\appendix

\bibliographystyle{ACM-Reference-Format}
\bibliography{paper}

\end{document}